\documentclass[11pt,a4paper]{article}

\usepackage[T1]{fontenc}
\usepackage{authblk}
\usepackage{amsmath,amssymb,amsthm,mathtools,mathrsfs}
\numberwithin{equation}{section}
\IfFileExists{newtxtext.sty}{%
  \usepackage{newtxtext,newtxmath}%
}{%
  \usepackage{mathptmx}%
}
\usepackage{microtype}
\usepackage[a4paper,margin=28mm,headheight=14pt]{geometry}
\usepackage{bm}
\usepackage{enumitem}
\usepackage{aliascnt}
\usepackage{booktabs}
\usepackage{float}
\usepackage{array}
\usepackage{tabularx}
\usepackage{longtable}
\usepackage{xcolor}
\usepackage{fancyhdr}
\usepackage{hyperref}
\usepackage[nameinlink,noabbrev]{cleveref}
\usepackage{url}

\hypersetup{
  colorlinks=true,
  linkcolor=blue,
  citecolor=blue,
  urlcolor=blue,
  pdftitle={Regular-Zero Drinfeld--Sokolov Charges in Doubled Four-Dimensional Chern--Simons Theory},
  pdfauthor={WenTao Zeng, YiFan He, Xian Wu}
}

\allowdisplaybreaks
\setlist[itemize]{leftmargin=2em,itemsep=0.2em,topsep=0.3em}
\setlist[enumerate]{leftmargin=2.2em,itemsep=0.2em,topsep=0.3em}

\newtheorem{theorem}{Theorem}[section]
\theoremstyle{plain}
\newaliascnt{proposition}{theorem}
\newtheorem{proposition}[proposition]{Proposition}
\aliascntresetthe{proposition}
\crefname{proposition}{Proposition}{Propositions}
\Crefname{proposition}{Proposition}{Propositions}
\newaliascnt{lemma}{theorem}
\newtheorem{lemma}[theorem]{Lemma}
\aliascntresetthe{lemma}
\crefname{lemma}{Lemma}{Lemmas}
\Crefname{lemma}{Lemma}{Lemmas}
\newaliascnt{corollary}{theorem}

\aliascntresetthe{corollary}
\crefname{corollary}{Corollary}{Corollaries}
\Crefname{corollary}{Corollary}{Corollaries}
\newaliascnt{assumption}{theorem}
\newtheorem{assumption}[theorem]{Assumption}
\aliascntresetthe{assumption}
\crefname{assumption}{Assumption}{Assumptions}
\Crefname{assumption}{Assumption}{Assumptions}
\theoremstyle{definition}
\newaliascnt{definition}{theorem}
\newtheorem{definition}[theorem]{Definition}
\aliascntresetthe{definition}
\crefname{definition}{Definition}{Definitions}
\Crefname{definition}{Definition}{Definitions}
\newaliascnt{example}{theorem}
\newtheorem{example}[theorem]{Example}
\aliascntresetthe{example}
\crefname{example}{Example}{Examples}
\Crefname{example}{Example}{Examples}
\theoremstyle{remark}
\newaliascnt{remark}{theorem}
\newtheorem{remark}[theorem]{Remark}
\aliascntresetthe{remark}
\crefname{remark}{Remark}{Remarks}
\Crefname{remark}{Remark}{Remarks}

\newcommand{\fg}{\mathfrak g}
\newcommand{\fh}{\mathfrak h}
\newcommand{\ff}{\mathfrak f}
\newcommand{\fs}{\mathfrak s}

\newcommand{\cA}{\mathcal A}
\newcommand{\cK}{\mathcal K}
\newcommand{\cL}{\mathscr L}
\newcommand{\Ad}{\operatorname{Ad}}

\newcommand{\Res}{\operatorname*{res}}
\newcommand{\Tr}{\operatorname{Tr}}
\newcommand{\ev}{\operatorname{ev}}
\newcommand{\grad}{\nabla}
\newcommand{\dd}{\mathrm d}
\newcommand{\id}{\mathbf 1}
\newcommand{\weak}{\mathrel{\approx}}
\newcommand{\pair}[2]{\left\langle #1,#2\right\rangle}
\newcommand{\circleint}{\int_{S^1_R}}

\crefname{theorem}{Theorem}{Theorems}
\Crefname{theorem}{Theorem}{Theorems}
\crefname{section}{Section}{Sections}
\Crefname{section}{Section}{Sections}
\crefname{subsection}{Subsection}{Subsections}
\Crefname{subsection}{Subsection}{Subsections}
\crefname{subsubsection}{Subsubsection}{Subsubsections}
\Crefname{subsubsection}{Subsubsection}{Subsubsections}
\crefname{appendix}{Appendix}{Appendices}
\Crefname{appendix}{Appendix}{Appendices}
\crefname{equation}{Equation}{Equations}
\Crefname{equation}{Equation}{Equations}
\crefname{table}{Table}{Tables}
\Crefname{table}{Table}{Tables}
\crefname{figure}{Figure}{Figures}
\Crefname{figure}{Figure}{Figures}
\AtBeginDocument{}
\AtBeginDocument{}
\AtBeginDocument{}
\AtBeginDocument{}
\AtBeginDocument{}

\title{\bfseries Regular-Zero Drinfeld--Sokolov Charges\\[2mm]
       in Doubled Four-Dimensional Chern--Simons Theory}
\author{WenTao Zeng\textsuperscript{1}, YiFan He\textsuperscript{2}, and Xian Wu\textsuperscript{3}}
\affil[1,2,3]{School of Mathematical Sciences, Capital Normal University, BeiJing 100048, P.R.China}
\date{}
\begin{document}
\maketitle
\footnotetext[1]{\scriptsize\textsuperscript{1} E-mail: 2240502162@cnu.edu.cn;\quad \textsuperscript{2}E-mail:2260501029@cnu.edu.cn;\quad \textsuperscript{3}  E-mail:2240502163@cnu.edu.cn}

\begin{abstract}
We study classical local conserved charges obtained after finite-defect reduction of doubled four-dimensional Chern--Simons theory. At a finite simple zero, the inhomogeneous derivative term in the affine gauge transformation induced by the twist function vanishes, so the residue evaluations of the two Gaudin currents transform by the ordinary adjoint action. 
For each simple ideal, we use the Evans-compatible invariant tensors of the principal Drinfeld--Sokolov hierarchy as external algebraic input, construct local circle integrals, and treat the same-zero, distinct-zero, and mixed sectors separately, proving pairwise Poisson commutativity.

Under the explicit assumptions of Stedman's simple-pole/simple-zero boundary Hamiltonian, the periodic trivial-monodromy sector, and boundary weights independent of phase-space variables, the reduced Hamiltonian decomposes term by term into finite-pole, infinity, finite-zero, and residual first-class-constraint contributions, which yields conservation of the constructed charges. 
We distinguish the finite-defect reduced space, the extended space carrying the residual constraint, and the physical quotient. The result is a conditional existence statement for a classical local Abelian Poisson subalgebra; we make no claim of functional independence, maximality, or Liouville completeness. Extensions to higher-order zeros and nontrivial monodromy sectors are outside the scope of this work.
\end{abstract}

\vspace{1em}

\clearpage
{\small\tableofcontents}
\newpage

\section{Introduction}

Four-dimensional Chern--Simons theory embeds two-dimensional integrable structures with a spectral parameter into a topological field theory. The doubled construction retains two Gaudin currents and their associated Lax connections; after finite-defect constraint reduction, Stedman gives the phase space, residual first-class constraint, and reduced Hamiltonian \cite{Stedman2026}. 
We study how to construct local Poisson-commuting conserved charges at regular zeros in this reduced structure.

Let the two rational currents be
\begin{equation*}
  \Gamma(x,z)=\sum_{q\in P_N}\frac{\Gamma_q(x)}{z-q},
  \qquad
  \Delta(x,z)=\sum_{q\in P_N}\frac{\Delta_q(x)}{z-q}.
\end{equation*}
They strongly Poisson-commute, but both carry the affine action of the residual $H$ gauge constraint:
\begin{equation*}
  \delta_\mu\Gamma(x,z)
  =\frac{\varphi(z)}{2\pi}\partial_x\mu(x)+[\Gamma(x,z),\mu(x)],
  \qquad
  \delta_\mu\Delta(x,z)
  =\frac{\varphi(z)}{2\pi}\partial_x\mu(x)+[\Delta(x,z),\mu(x)].
\end{equation*}
Thus, at a generic spectral parameter, invariant polynomials still contain an inhomogeneous derivative term and do not directly define physical local charges.

Let $\zeta$ be a finite simple zero of the twist function, namely
\[
  \varphi(\zeta)=0,\qquad \varphi'(\zeta)\neq0.
\]
The corresponding normalized evaluation currents are
\begin{equation*}
  j_\zeta^{(g)}=\frac{2\pi}{\varphi'(\zeta)}\Gamma(\zeta),
  \qquad
  j_\zeta^{(h)}=\frac{2\pi}{\varphi'(\zeta)}\Delta(\zeta).
\end{equation*}
The zero condition removes the derivative term in the affine transformation, while simplicity keeps the normalization finite. Hence the circle integrals obtained by applying any adjoint-invariant polynomial to $j_\zeta^{(g)}$ or $j_\zeta^{(h)}$ are local strong physical charges.

Commutativity follows from the Poisson bracket of evaluation currents at the same point. Its general form is
\begin{equation*}
  \{j_1^a(x),j_2^b(y)\}
  =\alpha f^{ab}{}_{c}\widetilde j^c(x)\delta(x-y)
   +\beta\kappa^{ab}\delta'(x-y),
  \qquad \beta\neq0.
\end{equation*}
The first term vanishes because gradients of invariant polynomials commute; the circle integral of the second term is a total derivative precisely when the invariant tensors satisfy the Evans compatibility condition. The principal Drinfeld--Sokolov hierarchy supplies such compatible tensors indexed by affine exponents \cite{EvansHassanMacKayMountain1999,Evans2001,DrinfeldSokolov1985}. 
The $\delta'$ term for distinct zeros vanishes by the twist identity, while mixed brackets between the two Gaudin sectors vanish identically.

Under the assumptions of finite simple poles, finite simple zeros, periodic trivial monodromy, and phase-space-independent boundary weights, the charge $Q^A$ associated with each simple ideal, sector, and DS label satisfies
\begin{equation*}
  \{C_\infty[\mu],Q^A\}=0,
  \qquad
  \{Q^A,Q^B\}=0,
  \qquad
  \{Q^A,H_{\rm red}\}=0.
\end{equation*}
The last equality is strong for externally prescribed gauge multipliers; field-dependent gauge fixing is understood only on the constraint surface. Fixing the zero, sector, simple ideal, and DS index while varying only the hierarchy integer gives a countable candidate family; nonvanishing, functional independence, and completeness require separate analysis.

The paper is organized as follows. Section 2 gives the reduced Hamiltonian data and conventions; Section 3 derives the residual gauge action in the two Gaudin sectors; Section 4 discusses regular-zero evaluation and physicality; Section 5 derives the Evans criterion from same-point current brackets and connects it to the DS hierarchy; Section 6 completes the Poisson calculations for distinct zeros and mixed sectors; Section 7 proves conservation under the reduced Hamiltonian; 
and Section 8 collects explicit charge formulas and the scope of the result. Explicit $A$-type expansions, low-order calculations, and a nonvanishing witness are given in the appendices.

\section{Reduced Hamiltonian Data and Conventions}\label{sec:setup}

\subsection{Conventions}

Let $G$ and $H$ be connected complex semisimple Lie groups with Lie algebras $\fg$ and $\fh$, and assume $H\subset G$. Relative to a fixed nondegenerate invariant symmetric bilinear form $\langle\cdot,\cdot\rangle$ on $\fg$, assume the reductive orthogonal decomposition
\begin{equation}\label{eq:reductive}
  \fg=\fh\oplus\ff,
  \qquad [\fh,\ff]\subset\ff,
  \qquad \pair{\fh}{\ff}=0,
\end{equation}
We do not impose the symmetric-space condition $[\ff,\ff]\subset\fh$. Choose a basis $\{T_a\}$ of $\fg$ and its dual basis $\{T^a\}$, and write the split Casimir as
\begin{equation*}
  C^{\fg}_{12}=T_a\otimes T^a,
\end{equation*}
The corresponding object for $\fh$ is denoted $C^{\fh}_{12}$. We use the tensor-component notation $X_1=X\otimes\id$ and $X_2=\id\otimes X$ throughout.

For a polynomial $P\in S(\fg^*)$, define its gradient at $X\in\fg$ by the invariant form $\dd P_X(Y)=\pair{\grad P(X)}{Y}$. A symmetric tensor $d^{(m)}\in S^m(\fg^*)$ determines the homogeneous polynomial $P_{d^{(m)}}(X)=d^{(m)}(X,\ldots,X)$; parentheses on indices denote normalized symmetrization.

The reduced spatial manifold is the circle $S^1_R=\mathbb R/(2\pi R\mathbb Z)$. We work only in the periodic loop sector with trivial monodromy: all fields and gauge parameters are smooth and strictly periodic, so integration by parts produces no boundary terms. Nontrivial monodromy or twisted periodicity would require a new check of circle total-derivative terms and is not considered here.

To avoid confusion between reduction levels, let $\mathcal P_{\rm fin}$ be the Poisson space after finite-defect constraints have been imposed strongly, let $\mathcal P_{\rm ext}$ be the space still carrying the residual first-class constraint $C_\infty$, and let
\[
\mathcal P_{\rm phys}=N_{\mathcal P_{\rm ext}}(\mathcal I_\infty)/\mathcal I_\infty
\]
be the physical Poisson algebra obtained by quotienting by the residual constraint ideal $\mathcal I_\infty$. In the main text, a ``strong equality'' is first understood on $\mathcal P_{\rm ext}$, while a ``weak equality'' is understood after imposing $C_\infty\approx0$.

\subsection{Reduced Fields and the Residual First-Class Constraint}

Let $P_N\subset\mathbb C$ be the finite set of poles. The meromorphic one-form is
\begin{equation}\label{eq:twist-one-form}
  \omega=\varphi(z)\,\dd z,
  \qquad
  \varphi(z)=\sum_{q\in P_N}\frac{k_q}{z-q},
\end{equation}
where $k_q\in\mathbb C$ are the pole levels. The regular-zero construction uses the finite simple zeros of $\varphi$.

Following \cite{Stedman2026}, after imposing the finite-defect constraints the reduced variables may be taken as
\begin{equation*}
  (\Gamma_q(x),\gamma_q(x)),
  \qquad
  \Gamma_q:S^1_R\to\fg,
  \qquad
  \gamma_q:S^1_R\to G.
\end{equation*}
Their basic nontrivial brackets are
\begin{align}
  \{\Gamma_{q1}(x),\Gamma_{p2}(y)\}
  &=\delta_{qp}\left(
      [C^{\fg}_{12},\Gamma_{q1}(x)]\delta_{xy}
      -\frac{k_q}{2\pi}C^{\fg}_{12}\delta'_{xy}
    \right),\label{eq:gamma-pole-bracket}\\
  \{\Gamma_{q1}(x),\gamma_{p2}(y)\}
  &=\delta_{qp}\,\gamma_{q2}(y)C^{\fg}_{12}\delta_{xy}.
  \label{eq:gamma-group-bracket}
\end{align}
Define the $\fh$-valued field
\begin{equation}\label{eq:Delta-definition}
  \Delta_q
  :=\left(
       -\frac{k_q}{2\pi}\gamma_q'\gamma_q^{-1}
       +\gamma_q\Gamma_q\gamma_q^{-1}
     \right)^{\fh}.
\end{equation}
The reduced brackets derived in \cite{Stedman2026} are
\begin{align}
  \{\Gamma_{q1}(x),\Delta_{p2}(y)\}&=0,
  \label{eq:mixed-pole-zero}\\
  \{\Delta_{q1}(x),\Delta_{p2}(y)\}
  &=\delta_{qp}\left(
      -[C^{\fh}_{12},\Delta_{q1}(x)]\delta_{xy}
      +\frac{k_q}{2\pi}C^{\fh}_{12}\delta'_{xy}
    \right).
  \label{eq:delta-pole-bracket}
\end{align}
Thus the two current algebras have opposite signs and strongly Poisson-commute.

Introduce the rational Gaudin currents
\begin{equation*}
  \Gamma(x,z)=\sum_{q\in P_N}\frac{\Gamma_q(x)}{z-q},
  \qquad
  \Delta(x,z)=\sum_{q\in P_N}\frac{\Delta_q(x)}{z-q},
\end{equation*}
and the spatial Lax matrices
\begin{equation*}
  \cL_x^{(g)}(x,z)=\frac{2\pi}{\varphi(z)}\Gamma(x,z),
  \qquad
  \cL_x^{(h)}(x,z)=\frac{2\pi}{\varphi(z)}\Delta(x,z).
\end{equation*}

The constraint $C_\infty$ is read directly from these two current algebras:
\begin{equation}\label{eq:C-infinity}
  C_\infty(x)
  =\sum_{q\in P_N}\bigl(\Gamma_q^{\fh}(x)-\Delta_q(x)\bigr)
  \weak0,
\end{equation}
where the superscript $\fh$ denotes the orthogonal projection to $\fh$ along \eqref{eq:reductive}. Because the ultralocal coefficients in \eqref{eq:gamma-pole-bracket} and \eqref{eq:delta-pole-bracket} have opposite signs, the two $\delta'_{xy}$ terms cancel in the self-bracket of $C_\infty$, leaving only the adjoint action,
\begin{equation*}
  \{C_\infty(x),C_\infty(y)\}
  =[C^{\fh}_{12},C_\infty(x)]\,\delta_{xy}.
\end{equation*}
This equality holds as a distributional identity, including the contact term at $x=y$; the central $\delta'_{xy}$ terms from the two current algebras cancel exactly. Hence $C_\infty$ is first class.

Smear it with a smooth periodic test function $\mu:S^1_R\to\fh$,
\begin{equation}\label{eq:C-smeared}
  C_\infty[\mu]
  =\circleint\pair{\mu(x)}{C_\infty(x)}\,\dd x ,
\end{equation}
We use this form below. The symbol $\weak$ denotes weak equality in the Dirac sense \cite{HenneauxTeitelboim1992}, meaning that the equality holds only on the constraint surface; the constructed charges satisfy the stronger property that the bracket of \eqref{eq:C-smeared} with each charge vanishes before imposing the constraint.

The reduced data and constraint \eqref{eq:C-infinity} are taken from \cite{Stedman2026}; the four-dimensional symplectic form and the Dirac reduction of finite defects are not repeated here. All Poisson calculations are performed on $\mathcal P_{\rm fin}$, and $C_\infty$ is then treated as the residual first-class constraint on $\mathcal P_{\rm ext}$.

\section{Residual Gauge Action in the Two Gaudin Sectors}\label{sec:gauge}

The transformation in the $\Gamma$ sector follows from the reduced Hamiltonian analysis:
\begin{equation}\label{eq:Gamma-gauge-action}
  \{C_\infty[\mu],\Gamma(x,z)\}
  =\frac{\varphi(z)}{2\pi}\mu'(x)+[\Gamma(x,z),\mu(x)].
\end{equation}
We now derive the corresponding formula for $\Delta$ directly from the reduced brackets. Fix a pole $p$. By \eqref{eq:mixed-pole-zero}, $\{\Gamma_q,\Delta_p\}=0$, so only the $-\Delta_q$ terms in \eqref{eq:C-infinity} contribute. Contract the pole bracket \eqref{eq:delta-pole-bracket} with $\mu_1(x)$ and integrate the $\delta'_{xy}$ term by parts:

\begin{align}
  \{C_\infty[\mu],\Delta_{p2}(y)\}
  &=-\int\dd x\,\pair{\mu_1(x)}{\{\Delta_{p1}(x),\Delta_{p2}(y)\}}
  \nonumber\\
  &=-\int\dd x\,\pair{\mu_1(x)}{
      -[C^{\fh}_{12},\Delta_{p1}(x)]\delta_{xy}
      +\frac{k_p}{2\pi}C^{\fh}_{12}\delta'_{xy}}
  \nonumber\\
  &=[\Delta_p(y),\mu(y)]+\frac{k_p}{2\pi}\mu'(y).
  \label{eq:Delta-pole-gauge}
\end{align}

Multiplying by $(z-p)^{-1}$ and summing over poles, the second term reorganizes into $\varphi(z)$ by \eqref{eq:twist-one-form}, hence

\begin{equation}\label{eq:Delta-gauge-action}
  \{C_\infty[\mu],\Delta(x,z)\}
  =\frac{\varphi(z)}{2\pi}\mu'(x)+[\Delta(x,z),\mu(x)].
\end{equation}

The current algebra of $\Delta$ has the opposite sign to that of $\Gamma$, but the affine derivative terms are identical: the $+k_p/2\pi$ in \eqref{eq:Delta-pole-gauge} comes from the ultralocal coefficient with the opposite sign in \eqref{eq:delta-pole-bracket}, and the two sign reversals cancel. This is why the two sectors fit into the same affine $H$-gauge action. 
Combining \eqref{eq:Gamma-gauge-action} and \eqref{eq:Delta-gauge-action} and multiplying by $2\pi/\varphi(z)$ gives
\begin{equation*}
  \{C_\infty[\mu],\cL_x^{(s)}(x,z)\}
  =\partial_x\mu(x)+[\cL_x^{(s)}(x,z),\mu(x)],
  \qquad s\in\{g,h\}.
\end{equation*}
Thus both Lax matrices transform as spatial gauge connections. At a generic value of $z$, however, polynomials in $\Gamma(z)$ or $\Delta(z)$ are not invariant because of the derivative term.

\section{Regular-Zero Evaluation and Physical Local Observables}\label{sec:zeros}

Define
\begin{equation*}
  Z_{\mathrm{fin}}
  =\{\zeta\in\mathbb C\setminus P_N\mid\varphi(\zeta)=0\}.
\end{equation*}
We restrict to finite simple zeros.

\begin{assumption}[Regular simple zeros]\label{ass:simple-zero}
For every $\zeta\in Z_{\mathrm{fin}}$ considered below,
\begin{equation*}
  \varphi(\zeta)=0,
  \qquad \varphi'(\zeta)\neq0.
\end{equation*}
Since $\zeta\notin P_N$, the currents $\Gamma(z)$ and $\Delta(z)$ are regular at $z=\zeta$.
\end{assumption}

Near such a zero,
\begin{equation*}
  \varphi(z)=\varphi'(\zeta)(z-\zeta)+O((z-\zeta)^2).
\end{equation*}
Consequently, each spatial Lax matrix has at most a simple pole at $z=\zeta$; if the corresponding $\Gamma(\zeta)$ or $\Delta(\zeta)$ vanishes, the pole degenerates and disappears.

\begin{definition}[Regular-Zero Residue Currents]\label{def:zero-currents}
For $\zeta\in Z_{\mathrm{fin}}$, set
\begin{align}
  j_\zeta^{(g)}(x)
  &:=\Res_{z=\zeta}\cL_x^{(g)}(x,z)\,\dd z
    =\frac{2\pi}{\varphi'(\zeta)}\Gamma(x,\zeta),
  \label{eq:jg-def}\\
  j_\zeta^{(h)}(x)
  &:=\Res_{z=\zeta}\cL_x^{(h)}(x,z)\,\dd z
    =\frac{2\pi}{\varphi'(\zeta)}\Delta(x,\zeta).
  \label{eq:jh-def}
\end{align}
\end{definition}

The role of regular-zero evaluation is now clear. Evaluating \eqref{eq:Gamma-gauge-action} and \eqref{eq:Delta-gauge-action} at $z=\zeta$, the inhomogeneous term $\frac{\varphi(z)}{2\pi}\mu'(x)$ vanishes because $\varphi(\zeta)=0$, and both equations reduce to the ordinary adjoint action. Multiplying by the nonzero scalar $2\pi/\varphi'(\zeta)$ gives the transformation law of the residue currents

\begin{align}
  \{C_\infty[\mu],j_\zeta^{(g)}(x)\}
  &=[j_\zeta^{(g)}(x),\mu(x)],
  \label{eq:jg-adjoint}\\
  \{C_\infty[\mu],j_\zeta^{(h)}(x)\}
  &=[j_\zeta^{(h)}(x),\mu(x)].
  \label{eq:jh-adjoint}
\end{align}

The vanishing of the affine term uses only $\varphi(\zeta)=0$, whereas the simple-zero condition $\varphi'(\zeta)\neq0$ makes the normalization $2\pi/\varphi'(\zeta)$ finite, so the residue is a well-defined local current rather than a divergent quantity. Evaluation therefore reduces affine gauge covariance to ordinary adjoint covariance, which we denote schematically by

\begin{equation*}
  \ev_\zeta:\quad
  \delta^{\mathrm{aff}}_\mu
  \longmapsto
  \delta^{\mathrm{ad}}_\mu .
\end{equation*}

The starting point is that the invariant-polynomial problem in the constrained theory becomes the ordinary invariant theory of a finite-dimensional Lie algebra after evaluation at a zero. We therefore do not need to classify the full algebra of affine $H$-invariant differential polynomials. We repeatedly use the standard fact that $[\grad P(X),X]=0$ for an adjoint-invariant polynomial, so the gradient lies in the center of the centralizer $\fs^X$; two gradients in the same centralizer therefore commute.

\begin{proposition}[Invariant densities are physical]\label{prop:invariant-physical}
Let $d^{(m)}\in S^m(\fg^*)^{\fg}$ and $e^{(n)}\in S^n(\fh^*)^{\fh}$. For each finite simple zero $\zeta$, the densities
\begin{align}
  \cK_{d,\zeta}^{(g)}(x)
  &=d^{(m)}(j_\zeta^{(g)}(x),\ldots,j_\zeta^{(g)}(x)),\\
  \cK_{e,\zeta}^{(h)}(x)
  &=e^{(n)}(j_\zeta^{(h)}(x),\ldots,j_\zeta^{(h)}(x))
\end{align}
commute pointwise with $C_\infty[\mu]$. Hence their circle integrals strongly Poisson-commute with the residual constraint.
\end{proposition}

\begin{proof}
By \eqref{eq:jg-adjoint}--\eqref{eq:jh-adjoint}, the residue currents obey the ordinary adjoint law $\delta_\mu j_\zeta=[j_\zeta,\mu]$. For a symmetric invariant tensor $d^{(m)}$,
\begin{align*}
  \delta_\mu d^{(m)}(j_\zeta,\ldots,j_\zeta)
  &=m\,d^{(m)}([j_\zeta,\mu],j_\zeta,\ldots,j_\zeta)\\
  &=-m\,d^{(m)}([\mu,j_\zeta],j_\zeta,\ldots,j_\zeta)=0,
\end{align*}
where the last equality is infinitesimal adjoint invariance with all arguments equal. In the $g$ sector, $G$-invariance implies $H$-invariance; in the $h$ sector we use $H$-invariance directly. No use is made of $C_\infty\weak0$.
\end{proof}

\begin{remark}
The normalization of a single residue requires a simple zero; the affine term vanishes using only $\varphi(\zeta)=0$. Higher-order zeros require jet or Takiff current constructions and are outside this proposition.
\end{remark}

\section{Regular-Zero Charges and Evans Compatibility}\label{sec:ds}

Section 4 showed that at a finite simple zero the residual gauge transformation of the residue current reduces to the ordinary adjoint action. Thus physicality of local charges requires only invariant polynomials, while Poisson commutativity is controlled by the non-ultralocal term in the evaluation-current bracket. We first compute for general invariant polynomials and then insert the tensors supplied by the principal Drinfeld--Sokolov hierarchy.

\subsection{Candidate Charges, Evaluation Currents, and Physicality}\label{sec:charges}

Fix a finite simple zero $\zeta\in Z_{\rm fin}$. Write
\begin{equation*}
  J_\zeta^{(g)}(x)=\Gamma(x,\zeta),\qquad
  \widetilde J_\zeta^{(g)}(x)
  =-\left.\partial_z\Gamma(x,z)\right|_{z=\zeta},
\end{equation*}
and for the second sector write
\begin{equation*}
  J_\zeta^{(h)}(x)=\Delta(x,\zeta),\qquad
  \widetilde J_\zeta^{(h)}(x)
  =-\left.\partial_z\Delta(x,z)\right|_{z=\zeta}.
\end{equation*}
The evaluation currents use the regular-zero normalization
\begin{equation}\label{eq:zero-current-normalization}
  j_\zeta^{(A)}(x)
  =\frac{2\pi}{\varphi'(\zeta)}J_\zeta^{(A)}(x),
  \qquad
  \widetilde j_\zeta^{(A)}(x)
  =\frac{2\pi}{\varphi'(\zeta)}\widetilde J_\zeta^{(A)}(x),
  \qquad A=g,h.
\end{equation}
The simple-zero condition ensures that this normalization is nonzero.

Let $P\in S^m(\mathfrak s^*)^{\mathfrak s}$ be a homogeneous adjoint-invariant polynomial on the complex semisimple Lie algebra $\mathfrak s$. Define the regular-zero candidate charge
\begin{equation}\label{eq:g-charge}
  Q_{\zeta,P}^{(g)}
  =\circleint P\bigl(j_\zeta^{(g)}(x)\bigr)\,\dd x,
\end{equation}
and in the $h$ sector define
\begin{equation}\label{eq:h-charge}
  Q_{\zeta,P}^{(h)}
  =\circleint P\bigl(j_\zeta^{(h)}(x)\bigr)\,\dd x.
\end{equation}
If $P$ is chosen on one simple ideal, extend it by zero on the complementary ideals. Since $j_\zeta^{(g)}$ is a rational linear combination of the $\Gamma_q$, whereas $j_\zeta^{(h)}$ is
\begin{equation*}
  \frac{2\pi}{\varphi'(\zeta)}
  \sum_{q\in P_N}\frac{1}{\zeta-q}
  \left(
    -\frac{k_q}{2\pi}\gamma_q'\gamma_q^{-1}
    +\gamma_q\Gamma_q\gamma_q^{-1}
  \right)^{\fh},
\end{equation*}
These densities involve only $\Gamma_q,\gamma_q,\gamma_q^{-1},\gamma_q'$ at the same point and are therefore local differential expressions.

At a zero, equations \eqref{eq:jg-adjoint}--\eqref{eq:jh-adjoint} give
\[
  \delta_\mu j_\zeta^{(A)}=[j_\zeta^{(A)},\mu].
\]
Hence, for any invariant polynomial $P$,
\[
  \delta_\mu P(j_\zeta^{(A)})
  =\pair{\grad P(j_\zeta^{(A)})}{[j_\zeta^{(A)},\mu]}=0.
\]
Therefore
\begin{equation}\label{eq:charges-physical}
  \{C_\infty[\mu],Q_{\zeta,P}^{(A)}\}=0
\end{equation}
On the extended Poisson space this is a strong equality. Physicality and Evans compatibility are distinct issues: the former uses only ordinary adjoint covariance at the zero, while the latter controls the $\delta'$ term in the same-point current bracket.

\subsection{The Same-Zero Current Bracket}\label{sec:same-zero-calculation}

For the first sector, summing the pole brackets gives
\begin{equation*}
  \{J_{\zeta1}^{(g)}(x),J_{\zeta2}^{(g)}(y)\}
  =[C^{\fg}_{12},\widetilde J_{\zeta1}^{(g)}(x)]\delta_{xy}
  -\frac{1}{2\pi}\sum_{q\in P_N}\frac{k_q}{(\zeta-q)^2}C^{\fg}_{12}\delta'_{xy}.
\end{equation*}
From
\begin{equation*}
  \varphi'(\zeta)=-\sum_{q\in P_N}\frac{k_q}{(\zeta-q)^2}
\end{equation*}
we obtain
\begin{equation}\label{eq:same-zero-g-bracket}
  \{J_{\zeta1}^{(g)}(x),J_{\zeta2}^{(g)}(y)\}
  =[C^{\fg}_{12},\widetilde J_{\zeta1}^{(g)}(x)]\delta_{xy}
  +\frac{\varphi'(\zeta)}{2\pi}C^{\fg}_{12}\delta'_{xy}.
\end{equation}
The second sector gives the overall sign reversal:
\begin{equation}\label{eq:same-zero-h-bracket}
  \{J_{\zeta1}^{(h)}(x),J_{\zeta2}^{(h)}(y)\}
  =-[C^{\fh}_{12},\widetilde J_{\zeta1}^{(h)}(x)]\delta_{xy}
  -\frac{\varphi'(\zeta)}{2\pi}C^{\fh}_{12}\delta'_{xy}.
\end{equation}

After normalization, the two sectors can be written uniformly as
\begin{equation}\label{eq:generic-zero-current-bracket}
  \{j_1^a(x),j_2^b(y)\}
  =\alpha f^{ab}{}_{c}\,\widetilde j^c(x)\delta_{xy}
   +\beta\kappa^{ab}\delta'_{xy},
  \qquad \beta\neq0,
\end{equation}
Here $\alpha$ and $\beta$ depend only on the overall sector sign and the chosen normalization, and $\widetilde j$ need not equal $j$.

Let $P$ and $R$ be invariant polynomials of degrees $m$ and $n$, respectively, and write
\begin{equation*}
  A_a(x)=\frac{\delta P(X)}{\delta j^a(x)},
  \qquad
  B_b(x)=\frac{\delta R(X)}{\delta j^b(x)}.
\end{equation*}
Invariance gives $[\grad P(X),X]=[\grad R(X),X]=0$. On the Zariski-dense set of regular semisimple $X$, both gradients lie in the same Cartan centralizer, and hence
\begin{equation*}
  [\grad P(X),\grad R(X)]=0.
\end{equation*}
This is a polynomial identity and therefore holds for all $X$. Thus the ultralocal part of the bracket is
\begin{equation}\label{eq:Evans-delta-part}
  \{Q_P,Q_R\}_{\delta}
  =\alpha\circleint
    \pair{[\grad P(j),\grad R(j)]}{\widetilde j}\,\dd x=0.
\end{equation}
A complete derivation of the functional variation, gradient commutativity, and distributional integration is given in Appendix \ref{app:ultralocal-derivation}.

After integrating over $y$, the non-ultralocal part becomes
\begin{align}
  \{Q_P,Q_R\}_{\delta'}
  &=\beta mn(n-1)\circleint
     S_{i_1\cdots i_{N-1}\mid c}
     j^{i_1}\cdots j^{i_{N-1}}\partial_xj^c\,\dd x,
     \label{eq:Evans-delta-prime}
\end{align}
Here $S$ is the coefficient obtained by contracting the tensors of $P$ and $R$. Let $N=m+n-2$ and symmetrize all undifferentiated indices; the Evans condition reads
\begin{equation}\label{eq:Evans-compatibility}
  S_{(i_1\cdots i_{N-1}\mid c)}
  =S_{(i_1\cdots i_{N-1}c)}.
\end{equation}
In the original tensor indices, this is
\begin{equation*}
  d^{(m)}_{(a_1\cdots a_{m-1}}{}^{c}
  e^{(n)}_{b_1\cdots b_{n-2})b_{n-1}c}
  =
  d^{(m)}_{(a_1\cdots a_{m-1}}{}^{c}
  e^{(n)}_{b_1\cdots b_{n-2}b_{n-1})c}.
\end{equation*}

Condition \eqref{eq:Evans-compatibility} is sufficient for \eqref{eq:Evans-delta-prime} to be a total derivative. Conversely, if $\beta\neq0$ and this bracket vanishes for every periodic current $j$, then the polynomial one-form
\begin{equation*}
  \omega
  =S_{i_1\cdots i_{N-1}\mid c}
   j^{i_1}\cdots j^{i_{N-1}}\,\dd j^c
\end{equation*}
Its integral along every closed curve is zero, so $\omega$ is closed. The closedness condition is exactly \eqref{eq:Evans-compatibility}. Thus, on the unrestricted formal local current algebra, the Evans condition is both sufficient and necessary for commutativity of same-point charges.

When \eqref{eq:Evans-compatibility} holds, there is a fully symmetric tensor $T_{i_1\cdots i_N}$ such that
\begin{equation}\label{eq:Evans-total-derivative}
  S_{i_1\cdots i_{N-1}\mid c}
  j^{i_1}\cdots j^{i_{N-1}}\partial_xj^c
  =\frac1N\partial_x
  \left(T_{i_1\cdots i_N}j^{i_1}\cdots j^{i_N}\right).
\end{equation}
Its integral vanishes under periodic boundary conditions. This completes the general same-zero charge calculation.

\subsection{Evans Tensors from the DS Hierarchy}\label{sec:ds-input}

Let $\fs$ be a complex simple Lie algebra of rank $r_\fs$, with exponents $m_1,\ldots,m_{r_\fs}$ and Coxeter number $h_\fs$. The highest-weight, derivative-free Hamiltonian densities of the principal Drinfeld--Sokolov hierarchy give Weyl-invariant tensors on the Cartan subalgebra, with degrees
\begin{equation}\label{eq:affine-exponents}
  s_{a,r}=m_a+r h_\fs,
  \qquad d_{a,r}=s_{a,r}+1.
\end{equation}
Evans' DS analysis shows that these Cartan tensors pairwise satisfy the Cartan version of \eqref{eq:Evans-compatibility}; their current brackets and compatibility condition are the core calculation in Evans' construction of the DS hierarchy. By Chevalley restriction and the Evans lemma under this restriction, they extend to
\begin{equation*}
  d_\fs^{(a,r)}\in
  \left(S^{d_{a,r}}\fs^*\right)^\fs
\end{equation*}
and continue to satisfy the same compatibility condition. We use only the highest derivative-free part, without identifying the full derivative-containing DS Hamiltonian density with a regular-zero charge density.

For a complex semisimple Lie algebra
\begin{equation*}
  \fg=\bigoplus_{I\in\mathcal I_g}\fg_I,
  \qquad
  \fh=\bigoplus_{J\in\mathcal I_h}\fh_J,
\end{equation*}
On each simple ideal take the tensors above and extend them by zero to the complementary ideals. The invariant form is orthogonal between distinct simple ideals and the structure constants have no cross-components, so the Evans condition holds ideal by ideal. We may therefore define
\begin{align*}
  Q_{\zeta,I,a,r}^{(g)}&:=Q_{\zeta,d_I^{(a,r)}}^{(g)},\\
  Q_{\zeta,J,b,r}^{(h)}&:=Q_{\zeta,e_J^{(b,r)}}^{(h)}.
\end{align*}
By \eqref{eq:g-charge}--\eqref{eq:h-charge}, these charges remain local and strongly physical; by the same-zero calculation above, they pairwise Poisson-commute at fixed zero and fixed sector. Repeated exponents retain the principal Heisenberg label, but we make no claim that charges of the same degree are functionally independent.

\section{Poisson Brackets of DS Charges}\label{sec:involution}

Section 5 established Poisson commutativity of DS charges at a common zero within each sector using the Evans compatibility condition. This section treats the two remaining cases: distinct zeros within one sector and charges from different sectors. At distinct zeros, the twist-function identity removes the $\delta'$ term and adjoint invariance cancels the remaining contribution; for mixed sectors, the current brackets vanish identically. Together with the same-zero calculation, these results complete the proof of pairwise Poisson commutativity of the DS charges.

\subsection{Distinct Zeros in the Same Sector}

We first compute the bracket of the Gaudin currents before imposing the zero conditions. For $z,w\notin P_N$ with $z\neq w$, their definitions and \eqref{eq:gamma-pole-bracket} give
\begin{align*}
  \{\Gamma_1(x,z),\Gamma_2(y,w)\}
  &=\sum_{q,p\in P_N}
    \frac{\{\Gamma_{q1}(x),\Gamma_{p2}(y)\}}{(z-q)(w-p)}\\
  &=\sum_{q\in P_N}\frac{1}{(z-q)(w-q)}
    \left([C^{\fg}_{12},\Gamma_{q1}(x)]\delta_{xy}
      -\frac{k_q}{2\pi}C^{\fg}_{12}\delta'_{xy}\right).
\end{align*}
Here the Kronecker delta $\delta_{qp}$ in the pole bracket reduces the double sum to a single sum. Using
\begin{equation*}
  \frac{1}{(z-q)(w-q)}
  =\frac{1}{w-z}\left(\frac{1}{z-q}-\frac{1}{w-q}\right),
\end{equation*}
we obtain
\begin{align*}
  \sum_{q\in P_N}\frac{\Gamma_q(x)}{(z-q)(w-q)}
  &=\frac{\Gamma(x,z)-\Gamma(x,w)}{w-z},\\
  \sum_{q\in P_N}\frac{k_q}{(z-q)(w-q)}
  &=\frac{\varphi(z)-\varphi(w)}{w-z}.
\end{align*}
Thus the full bracket is
\begin{align*}
  \{\Gamma_1(x,z),\Gamma_2(y,w)\}
  =\left[C^{\fg}_{12},
      \frac{\Gamma_1(x,z)-\Gamma_1(x,w)}{w-z}
    \right]\delta_{xy}-\frac{1}{2\pi}\frac{\varphi(z)-\varphi(w)}{w-z}C^{\fg}_{12}\delta'_{xy}.
\end{align*}
Now take $z=\zeta$ and $w=\eta$, where $\zeta,\eta\in Z_{\mathrm{fin}}$ and $\zeta\neq\eta$. Since both are zeros of the twist function,
\begin{equation*}
  \frac{\varphi(\zeta)-\varphi(\eta)}{\eta-\zeta}=0.
\end{equation*}
The $\delta'_{xy}$ term therefore vanishes, leaving
\begin{equation}\label{eq:distinct-zero-g-bracket}
  \{\Gamma_1(x,\zeta),\Gamma_2(y,\eta)\}
  =\left[C^{\fg}_{12},
      \frac{\Gamma_1(x,\zeta)-\Gamma_1(x,\eta)}{\eta-\zeta}
    \right]\delta_{xy}.
\end{equation}

We next compute the bracket of the charges in the notation of \eqref{eq:g-charge}:
\begin{equation*}
  Q_{\zeta,P}^{(g)}=\circleint P(j_\zeta^{(g)}(x))\,\dd x,
  \qquad
  Q_{\eta,R}^{(g)}=\circleint R(j_\eta^{(g)}(x))\,\dd x.
\end{equation*}
Let $P,R$ be adjoint-invariant polynomials, degrees of homogeneity are m and n, and denote their Lie-algebra-valued gradients by
\begin{equation*}
  A(x)=\grad P(j_\zeta^{(g)}(x)),\qquad
  B(x)=\grad R(j_\eta^{(g)}(x)).
\end{equation*}
The chain rule contracts the current bracket with $A(x)$ in the first tensor factor and $B(y)$ in the second. To evaluate this contraction, expand $C^{\fg}_{12}=\sum_aT_a\otimes T^a$. For any $U\in\fg$,
\begin{align*}
  [C^{\fg}_{12},U_1]&=\sum_a[T_a,U]\otimes T^a,\\
  \sum_a\pair{A}{[T_a,U]}\pair{B}{T^a}
  &=\pair{A}{[B,U]}=\pair{[A,B]}{U}.
\end{align*}
The second line uses $\sum_a\pair{B}{T^a}T_a=B$ and invariance of the bilinear form. Applying this identity to \eqref{eq:distinct-zero-g-bracket}, using the normalization \eqref{eq:zero-current-normalization}, and integrating the delta distribution over $y$ gives
\begin{equation*}
  \{Q_{\zeta,P}^{(g)},Q_{\eta,R}^{(g)}\}
  =\frac{(2\pi)^{2}}{\varphi'(\zeta)\varphi'(\eta)(\eta-\zeta)}
    \circleint\pair{[A(x),B(x)]}
      {\Gamma(x,\zeta)-\Gamma(x,\eta)}\,\dd x.
\end{equation*}
The two factors $2\pi/\varphi'(\zeta)$ and $2\pi/\varphi'(\eta)$ come from the two normalized residue currents; they are independent of the phase-space variables.

At each spatial point, adjoint invariance of $P$ implies, for every $\xi\in\fg$,
\begin{equation*}
  0=\dd P_{j_\zeta^{(g)}}([\xi,j_\zeta^{(g)}])
   =\pair{A}{[\xi,j_\zeta^{(g)}]}
   =\pair{[j_\zeta^{(g)},A]}{\xi}.
\end{equation*}
By nondegeneracy, $[j_\zeta^{(g)},A]=0$; similarly, $[j_\eta^{(g)},B]=0$. Since the normalization constants are nonzero, these identities also give $[\Gamma(x,\zeta),A(x)]=[\Gamma(x,\eta),B(x)]=0$. Consequently,
\begin{align*}
  \pair{[A,B]}{\Gamma(x,\zeta)-\Gamma(x,\eta)}=\pair{B}{[\Gamma(x,\zeta),A]}
    -\pair{A}{[B,\Gamma(x,\eta)]}=0,
\end{align*}
and therefore
\begin{equation*}
  \{Q_{\zeta,P}^{(g)},Q_{\eta,R}^{(g)}\}=0
    \qquad(\zeta\neq\eta).
\end{equation*}
The cancellation uses the commutation of each gradient with its own current, without assuming $[A,B]=0$. It holds pointwise, so neither integration by parts nor the Evans condition is needed for distinct zeros.

For the $h$ sector, \eqref{eq:delta-pole-bracket} gives the same calculation with $\Gamma,\fg$ replaced by $\Delta,\fh$ and an overall minus sign. Hence $\{Q_{\zeta,P}^{(h)},Q_{\eta,R}^{(h)}\}=0$ as well. Taking $P,R$ to be the DS polynomials completes the distinct-zero calculation in each sector.

\subsection{Mixed Gaudin Sectors}

The polewise relation \eqref{eq:mixed-pole-zero} gives
\begin{equation}\label{eq:mixed-current-zero}
  \{\Gamma_1(x,z),\Delta_2(y,w)\}=0,
\end{equation}
Consequently, brackets between charges in the $g$ and $h$ sectors vanish for every finite simple zero and every DS label. No further use of the Evans condition is needed.

Combining the same-zero calculation of Section 5, the distinct-zero calculation of this section, and strong commutativity of the mixed sectors covers every possible pair of DS charges. The next section combines this Poisson conclusion with the four-term decomposition of the standard reduced Hamiltonian.

\section{Reduced Hamiltonian and Conservation}\label{sec:hamiltonian}

This section translates Stedman's standard boundary Hamiltonian into our notation and checks conservation of regular-zero DS charges term by term. For any charge $Q^A$ constructed above, the goal is to prove
\begin{equation}\label{eq:conservation-goal}
  \{Q^A,H_{\rm red}\}=0.
\end{equation}
Since these charges have no explicit time dependence, \eqref{eq:conservation-goal} is equivalent to
$\dd Q^A/\dd t=0$.

The standard reduced Hamiltonian decomposes as
\begin{equation*}
  H_{\rm red}
  =H_{\rm fin}+H_\infty+H_Z-C_\infty[a],
\end{equation*}
It is therefore enough to show that $Q^A$ has zero bracket with each of the four terms on the right. The only new calculations are for the finite-pole term $H_{\rm fin}$ and the infinity term $H_\infty$; the zero term $H_Z$ already belongs to the commuting DS algebra established above, while the constraint term is handled directly by the strong physicality proved earlier.

\subsection{Decomposition of the Physical Hamiltonian}\label{sec:physical-hamiltonian}

We now fix the precise definition of Stedman's reduced Hamiltonian in our notation and verify conservation. Unless stated otherwise, assume $k_q\neq0$, $K:=\sum_{q\in P_N}k_q\neq0$, all $\epsilon_q,\epsilon_\infty,\alpha_\zeta,\beta_\zeta$ are external constants independent of phase-space variables, and use the periodic trivial-monodromy sector of Section 2.

\begin{definition}[Residue Pairing at a Zero]\label{def:residue-pairing}
If $F(z),G(z)$ are meromorphic near $z=\zeta$, define
\[
\langle\!\langle F,G\rangle\!\rangle_{(\zeta)}
:=\operatorname*{Res}_{z=\zeta}
\varphi(z)\pair{F(z)}{G(z)}\,\dd z .
\]
Therefore
\[
\langle\!\langle \varphi^{-2}\Gamma,\Gamma\rangle\!\rangle_{(\zeta)}
:=\operatorname*{Res}_{z=\zeta}
\frac{\pair{\Gamma(z)}{\Gamma(z)}}{\varphi(z)}\,\dd z .
\]
\end{definition}

Using $\varphi(z)=d_\zeta(z-\zeta)+O((z-\zeta)^2)$ and regularity of $\Gamma(z)$ at $\zeta$ gives directly
\begin{equation}\label{eq:residue-zero-evaluation}
\operatorname*{Res}_{z=\zeta}
\frac{\pair{\Gamma(z)}{\Gamma(z)}}{\varphi(z)}\,\dd z
=\frac{\pair{\Gamma(\zeta)}{\Gamma(\zeta)}}{d_\zeta},
\qquad d_\zeta:=\varphi'(\zeta).
\end{equation}
$\Delta(z)$ has exactly the same formula.

\begin{table}[H]
\centering
\caption{Correspondence between Stedman's reduced Hamiltonian and our notation}
\begin{tabularx}{0.96\textwidth}{>{\raggedright\arraybackslash}p{0.30\textwidth} >{\raggedright\arraybackslash}p{0.30\textwidth} X}
\toprule
Stedman notation & Our notation & Role\\
\midrule
$\Gamma_q$ & $\Gamma_q$ & First Gaudin sector\\
$\Delta_q$ & $\Lambda_q^{\fh}$ & Second Gaudin sector\\
$C_\infty$ & $\sum_q(\Gamma_q^{\fh}-\Delta_q)$ & Residual first-class constraint\\
$\mathcal N_0$ & $-a$ (with the sign convention of this section) & Constraint multiplier\\
$\langle\!\langle\cdot,\cdot\rangle\!\rangle_{(\zeta)}$ & \eqref{eq:residue-zero-evaluation} & Zero Hamiltonian term\\
\bottomrule
\end{tabularx}
\end{table}

Define
\begin{equation}\label{eq:lax-Lambda}
  \Lambda_q:=\gamma_q\Gamma_q\gamma_q^{-1}
      -c_q\gamma_q'\gamma_q^{-1},
  \qquad c_q=\frac{k_q}{2\pi},
  \qquad \Delta_q=\Lambda_q^{\fh}.
\end{equation}
and write
\[
J_q^L=c_q\gamma_q^{-1}\gamma_q',
\qquad J_q^R=c_q\gamma_q'\gamma_q^{-1},
\qquad \Lambda_q=\gamma_q(\Gamma_q-J_q^L)\gamma_q^{-1}.
\]

Stedman's finite constraints $C_q=0$ have already been imposed strongly on $\mathcal P_{\rm fin}$, so the linear constraint-multiplier term in (98) no longer appears; the residual $\mathcal N_0 C_\infty$ term is retained as $-C_\infty[a]$. Hence the finite-pole term is
\begin{equation}\label{eq:lax-Hfin}
  H_{\rm fin}
  =\sum_{q\in P_N}\frac{\pi\epsilon_q}{k_q}
    \circleint\pair{\Lambda_q^{\ff}}{\Lambda_q^{\ff}}\,\dd x.
\end{equation}
Let
\[
S:=\sum_{q\in P_N}\Gamma_q,\qquad
A_0:=\frac{2\pi}{K}S,\qquad k_\infty=-K,
\]
The infinity term is
\begin{equation}\label{eq:lax-Hinfty}
  H_\infty
  =-\frac{\pi\epsilon_\infty}{K}
    \circleint\pair{S^\ff}{S^\ff}\,\dd x.
\end{equation}

By \eqref{eq:residue-zero-evaluation} and \eqref{eq:jg-def}--\eqref{eq:jh-def}, the boundary contribution at a finite simple zero is
\begin{equation}\label{eq:lax-HZ}
  H_Z
  =\sum_{\zeta\in Z_{\rm fin}}\frac{d_\zeta}{4\pi}
    \circleint\left(
      \alpha_\zeta\pair{j_\zeta^{(g)}}{j_\zeta^{(g)}}
      -\beta_\zeta\pair{j_\zeta^{(h)}}{j_\zeta^{(h)}}
    \right)\dd x.
\end{equation}
Equivalently,
\begin{equation*}
H_Z
=\pi\sum_{\zeta\in Z_{\rm fin}}\frac1{d_\zeta}
\circleint\left(
\alpha_\zeta\pair{\Gamma(\zeta)}{\Gamma(\zeta)}
-\beta_\zeta\pair{\Delta(\zeta)}{\Delta(\zeta)}
\right)\dd x.
\end{equation*}
Define
\begin{equation*}
\widehat Q_\zeta^{(g)}
:=\frac12\circleint\pair{j_\zeta^{(g)}}{j_\zeta^{(g)}}\,\dd x,\qquad
\widehat Q_\zeta^{(h)}
:=\frac12\circleint\pair{j_\zeta^{(h)}}{j_\zeta^{(h)}}\,\dd x.
\end{equation*}
Then
\begin{equation}\label{eq:lax-HZ-quadratic}
  H_Z
  =\sum_{\zeta\in Z_{\rm fin}}\frac{d_\zeta}{2\pi}
    \left(\alpha_\zeta\widehat Q_\zeta^{(g)}
    -\beta_\zeta\widehat Q_\zeta^{(h)}\right).
\end{equation}

In summary, the standard reduced Hamiltonian used in this section is
\begin{equation}\label{eq:lax-Hred}
  H_{\rm red}=H_{\rm fin}+H_\infty+H_Z-C_\infty[a].
\end{equation}
Equation \eqref{eq:lax-Hred} is Stedman's equation (139) after imposing the finite constraints, setting $k_\infty=-K$, and taking $\mathcal N_0=-a$; we do not extend it to field-dependent boundary weights here.

\subsection{Commutation of the Finite-Pole Term with DS Charges}

The finite-pole term appears unrelated to zero charges, but we must check that it does not indirectly drive the zero currents through the group variables $\gamma_q$. The key strong identity is the following.

\begin{lemma}[Common Commutativity of the Finite-Pole Energy]\label{lem:Hfin-commutes}
In the finite-defect reduced Poisson algebra,
\begin{equation}\label{eq:lax-Gamma-Lambda-short}
  \{\Gamma_{q1}(x),\Lambda_{p2}(y)\}=0.
\end{equation}
In addition, for any $\fh$-valued test function $\eta$,
\begin{equation}\label{eq:Delta-Lambda-f-action}
  \{\Delta_q[\eta],\Lambda_q^{\ff}\}
  =[\eta,\Lambda_q^{\ff}].
\end{equation}
Therefore
\begin{equation}\label{eq:lax-Hfin-commutant}
  \{\Gamma(z),H_{\rm fin}\}=0,
  \qquad
  \{\Delta(z),H_{\rm fin}\}=0.
\end{equation}
\end{lemma}

\begin{proof}
By the basic brackets \eqref{eq:gamma-pole-bracket}--\eqref{eq:gamma-group-bracket}, the smeared current $\Gamma_q[\xi]$ generates
\begin{equation*}
  \delta_\xi\gamma_q=\gamma_q\xi,
  \qquad
  \delta_\xi\Gamma_q=[\Gamma_q,\xi]+c_q\xi'.
\end{equation*}
Substituting directly into \eqref{eq:lax-Lambda} gives
\begin{equation*}
 \delta_\xi(\gamma_q\Gamma_q\gamma_q^{-1})
 =c_q\gamma_q\xi'\gamma_q^{-1},
 \qquad
 \delta_\xi(\gamma_q'\gamma_q^{-1})
 =\gamma_q\xi'\gamma_q^{-1},
\end{equation*}
The two terms cancel exactly, yielding \eqref{eq:lax-Gamma-Lambda-short}.

On the other hand, $\Delta_q=\Lambda_q^{\fh}$ generates the adjoint action of $H$. Smearing $\Delta_q$ with $\eta$, the same-pole bracket in the sign convention of this section gives the transformation of the full composite current
\begin{equation*}
  \{\Delta_q[\eta],\Lambda_q\}
  =[\eta,\Lambda_q]-c_q\eta'.
\end{equation*}
The final term is the affine central term and is $\fh$-valued. Using $[\fh,\ff]\subset\ff$ and $\Lambda_q^{\fh}=\Delta_q$, projection to $\ff$ gives
\begin{equation*}
  \{\Delta_q[\eta],\Lambda_q^{\ff}\}
  =\bigl([\eta,\Lambda_q]-c_q\eta'\bigr)^{\ff}
  =[\eta,\Lambda_q^{\ff}],
\end{equation*}
This is precisely \eqref{eq:Delta-Lambda-f-action}. By invariance of the bilinear form,
\begin{equation*}
 \left\{\Delta_q[\eta],
 \tfrac12\pair{\Lambda_q^{\ff}}{\Lambda_q^{\ff}}\right\}
 =\pair{\Lambda_q^{\ff}}{[\eta,\Lambda_q^{\ff}]}=0.
\end{equation*}
Combining this with \eqref{eq:lax-Hfin} gives \eqref{eq:lax-Hfin-commutant}.
\end{proof}

Every DS charge depends only on $\Gamma(\zeta)$ or $\Delta(\zeta)$ at a given zero. Hence \eqref{eq:lax-Hfin-commutant} immediately implies
\begin{equation}\label{eq:Q-Hfin-zero}
  \{Q^A,H_{\rm fin}\}=0.
\end{equation}
This does not mean $H_{\rm fin}=0$; it generally still drives the group variables $\gamma_q$, but it does not change the two Gaudin currents used to build the DS charges.

\subsection{Infinity Term and Commutation with DS Charges}

By \eqref{eq:lax-Hinfty}, for each pole
\begin{equation*}
  \frac{\delta H_\infty}{\delta\Gamma_q}
  =-\epsilon_\infty A_0^{\ff}.
\end{equation*}
For each finite pole, substituting the pole current bracket \eqref{eq:gamma-pole-bracket} gives
\begin{equation*}
  \{\Gamma_q(x),H_\infty\}
  =[\Gamma_q(x),\epsilon_\infty A_0^{\ff}]
   +\frac{k_q}{2\pi}
    \partial_x(\epsilon_\infty A_0^{\ff}).
\end{equation*}
Multiplying by $(z-q)^{-1}$ and summing over poles, the first term reorganizes into $[\Gamma(z),\epsilon_\infty A_0^{\ff}]$, while the second term reorganizes, using \eqref{eq:twist-one-form}, into a derivative term proportional to $\varphi(z)$. Thus
\begin{equation*}
\begin{aligned}
  \{\Gamma(z),H_\infty\}
  &=\sum_{q\in P_N}\frac{1}{z-q}\{\Gamma_q,H_\infty\}\\
  &=[\Gamma(z),\epsilon_\infty A_0^{\ff}]
   +\frac{1}{2\pi}\left(\sum_{q\in P_N}\frac{k_q}{z-q}\right)
    \partial_x(\epsilon_\infty A_0^{\ff})\\
  &=[\Gamma(z),\epsilon_\infty A_0^{\ff}]
   +\frac{\varphi(z)}{2\pi}
    \partial_x(\epsilon_\infty A_0^{\ff}).
\end{aligned}
\end{equation*}
At a finite zero $\zeta$, $\varphi(\zeta)=0$, so
\begin{equation*}
  \{j_\zeta^{(g)},H_\infty\}
  =[j_\zeta^{(g)},\epsilon_\infty A_0^{\ff}].
\end{equation*}
If $P$ is adjoint-invariant, then
$[\grad P(X),X]=0$, and therefore,
\begin{equation*}
\begin{aligned}
 \{Q_{P,\zeta}^{(g)},H_\infty\}
 &=\circleint
   \pair{\grad P(j_\zeta^{(g)})}
         {[j_\zeta^{(g)},\epsilon_\infty A_0^{\ff}]}
   \,\dd x
 =0.
\end{aligned}
\end{equation*}
The second sector is simpler: $H_\infty$ depends only on $\Gamma_q$, while
$\{\Delta,\Gamma\}=0$, so
\begin{equation*}
  \{Q_{p,\zeta}^{(h)},H_\infty\}=0.
\end{equation*}
Hence, for every DS charge constructed above,
\begin{equation}\label{eq:Q-Hinf-zero}
  \{Q^A,H_\infty\}=0.
\end{equation}

\subsection{The Zero Hamiltonian Term Reduces to Quadratic DS Charges}

The zero term has already been written in \eqref{eq:lax-HZ-quadratic} as
\begin{equation*}
  H_Z
  =\sum_{\zeta\in Z_{\rm fin}}\frac{d_\zeta}{2\pi}
   \left(\alpha_\zeta\widehat Q_\zeta^{(g)}
   -\beta_\zeta\widehat Q_\zeta^{(h)}\right),
\end{equation*}
The coefficients $d_\zeta\alpha_\zeta/(2\pi)$ and $d_\zeta\beta_\zeta/(2\pi)$ are constants fixed by standard boundary data and independent of phase-space variables. It remains only to identify $\widehat Q$ with the degree-two member of the DS charges in Section 5.

Let $\fs$ be a simple ideal of $\fg$ or $\fh$. A simple Lie algebra has exactly one exponent equal to $1$, so the DS labels contain a unique degree-two position $(a_2,0)$ with $m_{a_2}=1$. In this subsection we normalize the Evans tensor at this position by
\begin{equation*}
  d_\fs^{(a_2,0)}(X,X)=P_{2,\fs}(X):=\frac12\pair{X}{X}.
\end{equation*}
For each simple ideal, take the projected current and define
\begin{equation*}
  Q_{\zeta,\fs,2}^{(s)}
  :=\circleint P_{2,\fs}\bigl(j_{\zeta,\fs}^{(s)}(x)\bigr)\,\dd x,
  \qquad s\in\{g,h\}.
\end{equation*}
This is the notation for $Q_{\zeta,\fs,a_2,0}^{(s)}$ defined in Section 5 with this normalization; another nonzero normalization changes it only by a numerical factor. Since distinct simple ideals of a semisimple Lie algebra are mutually orthogonal, the quadratic form in \eqref{eq:lax-HZ} decomposes ideal by ideal, and hence
\begin{equation*}
  \widehat Q_\zeta^{(g)}
  =\sum_{I\in\mathcal I_g}Q_{\zeta,\fg_I,2}^{(g)},
  \qquad
  \widehat Q_\zeta^{(h)}
  =\sum_{J\in\mathcal I_h}Q_{\zeta,\fh_J,2}^{(h)}.
\end{equation*}
Thus $H_Z$ is a linear combination of finitely many degree-two DS charges, rather than an independent functional requiring a new calculation here.

Take any DS charge $Q^A$ constructed above. The same-zero and distinct-zero calculations of Sections 5 and 6 show that it Poisson-commutes with each quadratic DS charge on every simple ideal. Bilinearity of the Poisson bracket then immediately gives
\begin{equation}\label{eq:Q-HZ-zero}
  \{Q^A,H_Z\}=0.
\end{equation}
This step uses the assumption that $\alpha_\zeta$ and $\beta_\zeta$ are external constants in the standard boundary Hamiltonian; if these coefficients were phase-space functions, additional Poisson-bracket terms would appear and require separate analysis. Thus the conservation proof does not need to repeat the same-zero, distinct-zero, and mixed-sector distributional calculations from Section 6; the complete commutativity established there is exactly the required input.

\subsection{The Residual Constraint Term Vanishes by Strong Physicality}

We proved above that every DS charge is a strongly physical observable: for every externally prescribed periodic $\fh$-valued function $\mu(x)$,
\begin{equation}\label{eq:Q-constraint-strong}
  \{C_\infty[\mu],Q^A\}=0.
\end{equation}
When $a$ is an external gauge multiplier independent of phase space, substituting $\mu=a$ immediately gives
\begin{equation}\label{eq:Q-gauge-zero}
  \{Q^A,-C_\infty[a]\}=0.
\end{equation}
If gauge fixing instead makes $a=a[\Phi]$ depend on phase-space fields $\Phi$, apply the Leibniz rule to the smeared constraint and use the strong equality \eqref{eq:Q-constraint-strong} for field-independent test functions to obtain
\begin{equation*}
\begin{aligned}
  \{Q^A,C_\infty[a]\}
  &=\circleint\pair{\{Q^A,a(x)\}}{C_\infty(x)}\,\dd x,\\
  \{Q^A,-C_\infty[a]\}
  &=-\circleint\pair{\{Q^A,a(x)\}}{C_\infty(x)}\,\dd x.
\end{aligned}
\end{equation*}
The right-hand side explicitly contains $C_\infty$ and therefore vanishes on the constraint surface. This distinction means that external gauge multipliers give strong conservation, whereas field-dependent gauge fixing gives conservation only weakly on the constraint surface.

\subsection{Conservation Theorem}

\begin{proposition}[Commutativity and Hamiltonian Conservation of Regular-Zero DS Charges]\label{prop:main-result}
Let $Q^A,Q^B$ be any DS charges constructed in Section 5. For the standard reduced Hamiltonian \eqref{eq:lax-Hred} of Section 7,
\begin{equation}\label{eq:charge-conservation}
  \{Q^A,Q^B\}=0,
  \qquad
  \{Q^A,H_{\rm red}\}=0,
  \qquad
  \frac{\dd Q^A}{\dd t}=0.
\end{equation}
Poisson commutativity is always a strong equality; Hamiltonian conservation is also strong for externally prescribed gauge multipliers and holds on $C_\infty\weak0$ under field-dependent gauge fixing.
\end{proposition}

\begin{proof}
Here $\{Q^A,Q^B\}=0$ follows from the same-zero, distinct-zero, and mixed-sector calculations in Sections 5 and 6. For Hamiltonian conservation, bilinearity of the Poisson bracket and \eqref{eq:lax-Hred} give
\begin{equation*}
\begin{aligned}
 \{Q^A,H_{\rm red}\}
 &=\{Q^A,H_{\rm fin}\}
  +\{Q^A,H_\infty\}
  +\{Q^A,H_Z\}
  -\{Q^A,C_\infty[a]\}.
\end{aligned}
\end{equation*}
The four terms vanish by
\eqref{eq:Q-Hfin-zero}, \eqref{eq:Q-Hinf-zero}, \eqref{eq:Q-HZ-zero}, and \eqref{eq:Q-gauge-zero}, respectively. Hence
$\{Q^A,H_{\rm red}\}=0$.
These charges have no explicit time dependence, so Hamilton's equation immediately gives
$\dd Q^A/\dd t=0$.
\end{proof}

\section{Explicit Charge Formulas, Physical Algebra, and Scope}\label{sec:results}

Sections 5 and 6 completed the charge construction and Poisson-bracket calculations, and Section 7 completed the term-by-term conservation check for the Hamiltonian. This section records only the explicit formulas, the position of the charges in the physical Poisson algebra, and the scope of the result; matrix expansions and nonvanishing calculations are given in the appendices.

\subsection{Charges in the Two Sectors}

Fix a finite simple zero $\zeta\in Z_{\rm fin}$. If $\fs\subset\fg$ is a simple ideal, denote its exponents by $m_1,\ldots,m_{r_\fs}$ and its Coxeter number by $h_\fs$. For each principal Heisenberg label $a\in\{1,\ldots,r_\fs\}$ and $r\in\mathbb Z_{\ge0}$, the DS--Evans input gives a nonzero symmetric adjoint-invariant tensor
\begin{equation*}
 d_{\fs}^{(a,r)}\in
 \bigl(S^{d_{a,r}}\fs^*\bigr)^{\fs},
 \qquad
 d_{a,r}=m_a+r h_\fs+1.
\end{equation*}
Project $j_\zeta^{(g)}$ to $\fs$ and denote the result by $j_{\zeta,\fs}^{(g)}$. By \eqref{eq:jg-def}, it is a linear combination of pole currents
\begin{equation*}
 j_{\zeta,\fs}^{(g)}(x)
 =\frac{2\pi}{\varphi'(\zeta)}
   \sum_{q\in P_N}\frac{\Gamma_q(x)|_\fs}{\zeta-q}.
\end{equation*}
The regular-zero DS charge in the first sector is therefore
\begin{equation}\label{eq:result-Qg-explicit}
 Q_{\zeta,\fs,a,r}^{(g)}
 =\circleint
 d_{\fs}^{(a,r)}\!\left(
 j_{\zeta,\fs}^{(g)}(x),\ldots,j_{\zeta,\fs}^{(g)}(x)
 \right)\,\dd x.
\end{equation}

If $\fs\subset\fh$ is a simple ideal of the second sector, then by \eqref{eq:Delta-definition} and \eqref{eq:jh-def}
\begin{equation*}
 j_{\zeta,\fs}^{(h)}(x)
 =\frac{2\pi}{\varphi'(\zeta)}
 \sum_{q\in P_N}\frac{1}{\zeta-q}
 \left(
   -\frac{k_q}{2\pi}\gamma_q'(x)\gamma_q(x)^{-1}
   +\gamma_q(x)\Gamma_q(x)\gamma_q(x)^{-1}
 \right)^{\!\fs},
\end{equation*}
Here $(\cdot)^{\fs}$ means first project to $\fh$ and then to the chosen simple ideal $\fs$. The charge in the second sector is
\begin{equation}\label{eq:result-Qh-explicit}
 Q_{\zeta,\fs,a,r}^{(h)}
 =\circleint
 d_{\fs}^{(a,r)}\!\left(
 j_{\zeta,\fs}^{(h)}(x),\ldots,j_{\zeta,\fs}^{(h)}(x)
 \right)\,\dd x.
\end{equation}
These expressions depend only on $\Gamma_q$, $\gamma_q$, $\gamma_q^{-1}$, and $\gamma_q'$ at the same point $x$, so they are local differential functionals; no new phase-space variables or spatially nonlocal integrals are introduced.

\subsection{A Countable Family from a Fixed DS Index Column}

Fix
\[
(\zeta,s,\fs,a),\qquad s\in\{g,h\},
\]
Let only $r=0,1,2,\ldots$ vary and define
\begin{equation}\label{eq:result-countable-family}
  \mathscr Q_{\zeta,s,\fs,a}
  :=\left\{Q_{\zeta,\fs,a,r}^{(s)}
  \ \middle|\ r\in\mathbb Z_{\ge0}\right\}.
\end{equation}
The label set is countable, and the degrees $d_{a,r}=m_a+r h_\fs+1$ grow strictly with $r$. By \cref{prop:main-result}, the members of \eqref{eq:result-countable-family} pairwise Poisson-commute and are conserved by the standard Hamiltonian flow of Section 7.

The infinitude of labels must be distinguished from nonvanishing of the functionals. The DS input gives only nonzero tensors; to assert a countably infinite family of nonzero charges, one must verify that $j_\zeta^{(s)}$ is not identically zero in the chosen sector. Appendix \ref{app:nonzero} gives a two-pole $\mathfrak{sl}_2$ witness. We do not claim functional independence, nor do we claim that charges with different $r$ satisfy no polynomial relations.

\subsection{Position in the Physical Poisson Algebra}

Let $\mathcal I_\infty$ be the constraint ideal generated by all $C_\infty[\mu]$; the physical Poisson algebra is
\[
\mathcal A_{\rm phys}
=N_{\mathcal P_{\rm ext}}(\mathcal I_\infty)/\mathcal I_\infty,
\qquad
N_{\mathcal P_{\rm ext}}(\mathcal I_\infty)
=\{F\mid\{F,\mathcal I_\infty\}\subset\mathcal I_\infty\}.
\]
For field-independent periodic test functions, we proved above that
\[
\{C_\infty[\mu],Q^A\}=0
\]
is a strong equality on $\mathcal P_{\rm ext}$. If $b=b[\Phi]$ is field-dependent, then
\[
\{Q^A,C_\infty[b]\}
=\circleint\pair{\{Q^A,b(x)\}}{C_\infty(x)}\,\dd x
\]
belongs to $\mathcal I_\infty$, so $Q^A$ descends to $\mathcal A_{\rm phys}$. Let $\mathcal A_{\rm DS}$ be the unital subalgebra of the physical quotient generated by the equivalence classes of these charges and their finite products. Then
\[
\mathcal A_{\rm DS}\subset\mathcal A_{\rm phys},
\qquad
\{\mathcal A_{\rm DS},\mathcal A_{\rm DS}\}=0.
\]
Here ``local'' modifies only the density of a generator; a finite product is generally not a single local circle integral.

\subsection{Scope of the Result}

The result addresses classical local Poisson-commuting conserved charges at regular finite simple zeros. Its inputs are Stedman's finite-defect reduced Poisson algebra, the residual first-class constraint, and the standard simple-pole/simple-zero boundary Hamiltonian; existence of the DS--Evans tensors is used as finite-dimensional external input one simple factor at a time. Compared with the regular-zero affine Gaudin construction \cite{LacroixMagroVicedo2017}, the doubled theory additionally requires handling the same residual $H$ constraint and two Gaudin sectors with opposite signs that strongly Poisson-commute, and the calculations above accomplish both tasks.

We do not identify $\cA_{\rm DS}$ with the full physical local algebra, and we do not prove that it is a maximal Poisson-commuting subalgebra or a Liouville-complete family. Higher-order zeros require a jet version of the regular-zero construction; higher-order poles change the defect phase space and introduce Takiff/Weil--jet data, both outside the simple spectral-data setting. A directly computable realization, the $A$-type DS--Evans determinant family for $\mathfrak{sl}_N$, is placed in Appendix B.

\section{Conclusion}\label{sec:conclusion}

We constructed local Poisson-commuting conserved charges in the reduced doubled four-dimensional Chern--Simons theory under three explicit restrictions: the twist function has only finite simple zeros, the spatial direction lies in the periodic trivial-monodromy loop sector, and the boundary Hamiltonian has Stedman's simple-pole/simple-zero form with phase-space-independent weights. At a regular zero,
\[
\delta_\mu\Gamma(x,\zeta)=[\Gamma(x,\zeta),\mu(x)],
\qquad
\delta_\mu\Delta(x,\zeta)=[\Delta(x,\zeta),\mu(x)],
\]
Thus invariant polynomials of the residue currents strongly Poisson-commute pointwise with the residual constraint. The DS--Evans input supplies a compatible family of invariant tensors; Section 5 establishes the Evans criterion, Section 6 computes all three classes of evaluation-current brackets and proves complete Poisson commutativity, and Section 7 gives Hamiltonian conservation through the four-term decomposition and term-by-term calculation.

The precise statement is as follows: for every simple ideal, finite simple zero, and label satisfying the DS--Evans input, the resulting charge is a strongly physical local conserved functional on $\mathcal P_{\rm ext}$ and descends to $\mathcal A_{\rm phys}$. Under the nonvanishing condition in Appendix \ref{app:nonzero}, fixing one DS index column yields a countable family of nonzero charges. We do not prove functional independence, identify $\mathcal A_{\rm DS}$ with the full physical local algebra, or claim Liouville completeness.

Higher-order zeros require jet/Takiff currents, higher-order poles alter the defect Poisson structure, and nontrivial monodromy changes circle total-derivative terms. These issues, together with quantization and further reduction of concrete real sigma models, are left for future work.
\section*{Acknowledgments}
We are grateful to Prof. Wei-Zhong Zhao for his helpful discussions.
This work is supported by the National Natural Science Foundation of China (No. 12375004). 
\appendix

\section{Index Form of the Invariant-Tensor Identities}\label{app:tensors}

Let $d^{(m)}_{a_1\cdots a_m}$ be fully symmetric. Infinitesimal adjoint invariance is
\begin{equation*}
  \sum_{i=1}^{m}
  f_{b a_i}{}^{c}
  d^{(m)}_{a_1\cdots a_{i-1}c a_{i+1}\cdots a_m}=0.
\end{equation*}
Equivalently,
\begin{equation*}
  d^{(m)}_{c(a_2\cdots a_m}f_{a_1)b}{}^c=0.
\end{equation*}
After raising one index with the invariant form, the Evans identity is
\begin{equation*}
  d^{(m)}_{(a_1\cdots a_{m-1}}{}^{c}
  e^{(n)}_{b_1\cdots b_{n-2})b_{n-1}c}
  =d^{(m)}_{(a_1\cdots a_{m-1}}{}^{c}
  e^{(n)}_{b_1\cdots b_{n-2}b_{n-1})c}.
\end{equation*}
The right-hand side is fully symmetric in all $m+n-2$ free indices. This is precisely the symmetry needed to convert the $\delta'$ contribution in \eqref{eq:Evans-total-derivative} into a total derivative.

\section{\texorpdfstring{$A$}{A}-Type DS--Evans Charges: An Explicit Realization}\label{app:typeA}

This appendix specializes the abstract DS--Evans tensor construction to $\mathfrak{sl}_N(\mathbb C)$. The purpose is not to add assumptions, but to give an example that can be expanded directly and substituted into the regular-zero currents; physicality, Poisson commutativity, and conservation remain consequences of the general proof in the main text.

For $\fs=\mathfrak{sl}_N(\mathbb C)$, the exponents and Coxeter number are
\begin{equation*}
  1,2,\ldots,N-1,
  \qquad h_{\mathfrak{sl}_N}=N.
\end{equation*}
The allowed polynomial degrees are therefore
\begin{equation}\label{eq:typeA-degrees}
  n\ge2,
  \qquad n-1\not\equiv0\pmod N.
\end{equation}

For an $\mathfrak{sl}_N$ factor in the $g$ sector, define
\begin{equation*}
  X_\zeta(x)=2\pi\Gamma(x,\zeta)\big|_{\mathfrak{sl}_N}
  =\left.\varphi(z)\cL_x^{(g)}(x,z)\right|_{z=\zeta,\mathfrak{sl}_N},
\end{equation*}
and
\begin{equation*}
  D_\zeta(u,x)=\det(\id-uX_\zeta(x)).
\end{equation*}
For each allowed degree $n$, let
\begin{equation*}
  p_n=\frac{n-1}{N}
\end{equation*}
and define the normalized density
\begin{equation}\label{eq:typeA-generating-density}
  \widehat{\cK}_{n,\zeta}^{(g)}(x)
  :=-\frac{n}{p_n}
      \left[D_\zeta(u,x)^{p_n}\right]_{u^n},
\end{equation}
where $[F(u)]_{u^n}$ denotes the coefficient of $u^n$ in the formal expansion at $u=0$.

\begin{proposition}[$A$-Type Generating Family]\label{prop:typeA-family}
For degrees satisfying \eqref{eq:typeA-degrees}, the tensors obtained by polarizing \eqref{eq:typeA-generating-density} form an Evans-compatible family. Substituting these tensors into \eqref{eq:result-Qg-explicit} or \eqref{eq:result-Qh-explicit} gives an explicit matrix realization of the abstract DS--Evans charges.
\end{proposition}

\begin{proof}
This is the determinant construction of \cite{EvansHassanMacKayMountain1999,Evans2001}, restated at regular zeros following \cite{LacroixMagroVicedo2017}. The exponent $p_n=(n-1)/N$ is exactly what makes the generated tensors satisfy the commutativity identity. Since $X_\zeta$ differs from the residue current only by a nonzero scalar, these functionals belong to the same homogeneous hierarchy.
\end{proof}

Let
\begin{equation*}
  J_k=\Tr X_\zeta^k.
\end{equation*}
Since $J_1=0$,
\begin{equation}\label{eq:det-log-expansion}
  D_\zeta(u)^{p_n}
  =\exp\left[-p_n\sum_{k\ge2}\frac{u^k}{k}J_k\right].
\end{equation}
Expanding \eqref{eq:det-log-expansion} gives
\begin{align}
  \widehat{\cK}_{2,\zeta}&=J_2,\nonumber\\
  \widehat{\cK}_{3,\zeta}&=J_3,\nonumber\\
  \widehat{\cK}_{4,\zeta}&=J_4-\frac{3}{2N}J_2^2,\nonumber\\
  \widehat{\cK}_{5,\zeta}&=J_5-\frac{10}{3N}J_2J_3,\label{eq:K5}\\
  \widehat{\cK}_{6,\zeta}&=J_6-\frac{5}{3N}J_3^2
     -\frac{15}{4N}J_2J_4+\frac{25}{8N^2}J_2^3.
  \label{eq:K6}
\end{align}
For fixed $N$, only degrees satisfying \eqref{eq:typeA-degrees} yield nonzero hierarchy members. If $n-1=rN$, then $p_n=r$ is an integer and $D_\zeta(u)^r$ has degree $rN=n-1$, so the coefficient of $u^n$ vanishes identically.

\begin{example}[$\mathfrak{sl}_2$]
For $\mathfrak{sl}_2$, the exponent modulo $2$ is $1$, so the allowed degrees are even. Every invariant polynomial on $\mathfrak{sl}_2$ is generated by $J_2$; hence the hierarchy consists of selected powers of $J_2$, normalized so that the integrated charges obey the compatibility hierarchy.
\end{example}

\begin{example}[$\mathfrak{sl}_3$]
For $\mathfrak{sl}_3$, the allowed degrees satisfy $n\equiv0,2\pmod3$. The first few are $2,3,5,6,8,9,\ldots$. Degree $4$ is absent, while degree $5$ corresponds to the corrected combination \eqref{eq:K5}.
\end{example}

For an $\mathfrak{sl}_N$ factor in the $\fh$ sector, make the replacement
\begin{equation*}
  X_\zeta=2\pi\Gamma(\zeta)
  \quad\longmapsto\quad
  Y_\zeta=2\pi\Delta(\zeta).
\end{equation*}
The same formulas then hold.

\subsection{Low-Order Expansion of the Determinant Generating Formula}\label{app:determinant}

For $X\in\mathfrak{sl}_N$,
\begin{equation*}
  \log\det(\id-uX)
  =\Tr\log(\id-uX)
  =-\sum_{k\ge2}\frac{u^k}{k}\Tr X^k.
\end{equation*}
Let $p=(n-1)/N$ and write $J_k=\Tr X^k$. Then
\begin{equation*}
  \det(\id-uX)^p
  =\exp\left[-p\left(
    \frac{u^2}{2}J_2+\frac{u^3}{3}J_3+\frac{u^4}{4}J_4+\cdots
  \right)\right].
\end{equation*}
For example, the coefficient of $u^4$ is
\begin{equation*}
  -\frac{p}{4}J_4+\frac{p^2}{8}J_2^2.
\end{equation*}
At $n=4$, $p=3/N$; multiplying by the normalization factor $-4/p$ in \eqref{eq:typeA-generating-density} gives
\begin{equation*}
  J_4-\frac{3}{2N}J_2^2.
\end{equation*}
The same expansion gives \eqref{eq:K5} and \eqref{eq:K6}.

\section{\texorpdfstring{Nonvanishing Witness: A Two-Pole $\mathfrak{sl}_2$ Example}{Nonvanishing Witness: A Two-Pole sl2 Example}}\label{app:nonzero}

Here we give a simple configuration satisfying the constraint for which the zero charge is nonzero. Take
\[
\fg=\fh=\mathfrak{sl}_2,\qquad \ff=0,
\]
Take two distinct finite poles $q_1\ne q_2$ with levels satisfying $k_1k_2(k_1+k_2)\ne0$. Then
\[
\varphi(z)=\frac{k_1}{z-q_1}+\frac{k_2}{z-q_2}
\]
There is a unique finite zero
\[
\zeta=\frac{k_1q_2+k_2q_1}{k_1+k_2},
\qquad
\varphi'(\zeta)=\frac{k_1+k_2}{(\zeta-q_1)(\zeta-q_2)}\ne0.
\]
Take any constant $X\in\mathfrak{sl}_2$ satisfying $\pair{X}{X}\ne0$ and set
\[
\Gamma_{q_1}(x)=X,\qquad \Gamma_{q_2}(x)=0,
\qquad
\gamma_{q_1}(x)=\gamma_{q_2}(x)=\id.
\]
By \eqref{eq:Delta-definition}, $\Delta_q=\Gamma_q$, so $C_\infty=0$ holds strongly, and
\[
j_\zeta^{(g)}=\frac{2\pi}{\varphi'(\zeta)(\zeta-q_1)}\,X\ne0.
\]
The invariant-polynomial algebra of $\mathfrak{sl}_2$ is generated by the quadratic Casimir. Hence every even-degree member of the standard DS--Evans family in this configuration is a nonzero power of $P_2(X)$ times a nonzero constant. In particular,
\[
Q_{\zeta,2}^{(g)}
=\frac12\circleint\pair{j_\zeta^{(g)}}{j_\zeta^{(g)}}\,\dd x
=\pi R\pair{j_\zeta^{(g)}}{j_\zeta^{(g)}}\ne0.
\]
This example proves only that nonvanishing can occur; it does not prove functional independence of charges at different levels.

\section{Detailed Derivation of the Ultralocal Charge Bracket}\label{app:ultralocal-derivation}

Equation \eqref{eq:Evans-delta-part} asserts that the $\delta$ term in the Poisson bracket of local integrals built from two adjoint-invariant polynomials vanishes identically. Starting from \eqref{eq:generic-zero-current-bracket}, this appendix expands the functional variation, Lie-algebra contraction, and gradient commutativity; the derivation does not require $\widetilde j=j$, Evans compatibility, or the constraint $C_\infty\weak0$.

Fix one sector and temporarily suppress the zero and sector labels. Let $\fs$ be the complex semisimple Lie algebra of that sector, choose a basis $\{T_a\}$, and write
\begin{equation*}
  j(x)=j^a(x)T_a,\qquad
  \widetilde j(x)=\widetilde j^a(x)T_a,\qquad
  \kappa_{ab}=\pair{T_a}{T_b},\qquad
  \kappa^{ac}\kappa_{cb}=\delta^a_b.
\end{equation*}
Here $\kappa$ is a nondegenerate complex bilinear invariant form, and the gradient definition involves no complex conjugation. For invariant polynomials of degrees $m$ and $n$, use the factorial-free normalization of the main text
\begin{equation*}
  P(X)=d^{(m)}_{a_1\cdots a_m}X^{a_1}\cdots X^{a_m},\qquad
  R(X)=e^{(n)}_{b_1\cdots b_n}X^{b_1}\cdots X^{b_n},
\end{equation*}
The corresponding local functionals are
\begin{equation*}
  Q_P[j]=\circleint P(j(x))\,\dd x,\qquad
  Q_R[j]=\circleint R(j(x))\,\dd x.
\end{equation*}
For every periodic variation $\eta(x)=\eta^a(x)T_a$,
\begin{equation}\label{eq:app58-functional-variation}
  \left.\frac{\dd}{\dd\varepsilon}Q_P[j+\varepsilon\eta]
  \right|_{\varepsilon=0}
  =\circleint \partial_aP(j(x))\eta^a(x)\,\dd x.
\end{equation}
The density contains no spatial derivatives of $j$, so no integration by parts is needed here. By the definition of the functional derivative,
\begin{align}
  A_a(x):=\frac{\delta P(x)}{\delta j^a(x)}
  &=m\,d^{(m)}_{a a_2\cdots a_m}
    j^{a_2}(x)\cdots j^{a_m}(x),\nonumber\\
  B_b(x):=\frac{\delta R(x)}{\delta j^b(x)}
  &=n\,e^{(n)}_{b b_2\cdots b_n}
    j^{b_2}(x)\cdots j^{b_n}(x).
  \label{eq:app58-functional-gradients}
\end{align}
Raising the index gives the Lie-algebra-valued gradients
\begin{equation}\label{eq:app58-raised-gradients}
  A^\sharp(x)=\kappa^{ab}A_b(x)T_a=\grad P(j(x)),\qquad
  B^\sharp(x)=\kappa^{ab}B_b(x)T_a=\grad R(j(x)).
\end{equation}
This agrees with the definition $\dd P_X(Y)=\pair{\grad P(X)}{Y}$.

First, invariance implies for any $X,\xi\in\fs$,
\begin{equation*}
  0=\left.\frac{\dd}{\dd t}P(\Ad_{\exp(t\xi)}X)\right|_{t=0}
   =\pair{\grad P(X)}{[\xi,X]}
   =\pair{[X,\grad P(X)]}{\xi}.
\end{equation*}
The last step uses invariance of the bilinear form. Nondegeneracy gives
\begin{equation}\label{eq:app58-centralizer}
  [X,\grad P(X)]=0,\qquad [X,\grad R(X)]=0.
\end{equation}
Thus both gradients lie in the centralizer of $X$, but for a general $X$ the centralizer need not be abelian, so a further argument is required.

The gradient of an invariant polynomial is equivariant under the adjoint action. For any element $g$ of the corresponding connected Lie group and any $Y\in\fs$,
\begin{align*}
  \pair{\grad P(\Ad_gX)}{Y}
  &=\left.\frac{\dd}{\dd\varepsilon}P(\Ad_gX+\varepsilon Y)
      \right|_{\varepsilon=0}\\
  &=\left.\frac{\dd}{\dd\varepsilon}P(X+\varepsilon\Ad_{g^{-1}}Y)
      \right|_{\varepsilon=0}\\
  &=\pair{\grad P(X)}{\Ad_{g^{-1}}Y}
   =\pair{\Ad_g\grad P(X)}{Y}.
\end{align*}
Nondegeneracy gives
\begin{equation}\label{eq:app58-gradient-equivariance}
  \grad P(\Ad_gX)=\Ad_g\grad P(X).
\end{equation}
Set $g=\exp(t\xi)$ and differentiate at $t=0$ to obtain the infinitesimal form
\begin{equation}\label{eq:app58-infinitesimal-equivariance}
  D(\grad P)_X\bigl([\xi,X]\bigr)=[\xi,\grad P(X)],
\end{equation}
Here $D(\grad P)_X$ is the linearization of the gradient map at $X$.

Now fix $X$ and take $\xi=\grad R(X)$. In this differentiation, $\xi$ is the constant determined by this fixed $X$ and does not vary with $t$. By \eqref{eq:app58-centralizer}, the direction $[\xi,X]$ on the left vanishes, so
\begin{equation*}
  [\grad R(X),\grad P(X)]
  =D(\grad P)_X\bigl([\grad R(X),X]\bigr)=0.
\end{equation*}
Hence for all $X\in\fs$,
\begin{equation}\label{eq:app58-gradient-commutator}
  [\grad P(X),\grad R(X)]=0.
\end{equation}
The argument applies equally to singular and nilpotent elements. Equivalently, \eqref{eq:app58-infinitesimal-equivariance} shows that $\grad P(X)$ commutes with every element of the centralizer of $X$, and $\grad R(X)$ is one such element.

The main text uses a shorter density argument: on the Zariski-dense set of regular semisimple elements, the centralizer is a Cartan subalgebra, so the two gradients commute; since the left-hand side of \eqref{eq:app58-gradient-commutator} is polynomial in $X$, the identity extends to all of $\fs$. The Cartan argument requires regular semisimple elements, whereas the equivariance proof above covers every $X$ directly.

Take the $\delta$ term in \eqref{eq:generic-zero-current-bracket} separately. The chain rule for functional Poisson brackets gives
\begin{align}
  \{Q_P,Q_R\}_{\delta}
  &=\alpha\int_{S^1_R}\dd x\int_{S^1_R}\dd y\,
      A_a(x)B_b(y)f^{ab}{}_{c}\widetilde j^c(x)\delta_{xy}
      \nonumber\\
  &=\alpha\circleint
      A_a(x)B_b(x)f^{ab}{}_{c}\widetilde j^c(x)\,\dd x.
  \label{eq:app58-delta-integration}
\end{align}
The second line uses only the evaluation property of the periodic $\delta$ distribution,
$\int_{S^1_R}B_b(y)\delta_{xy}\,\dd y=B_b(x)$, with no integration by parts and no boundary term.

To check the indices and signs in the contraction, write
\begin{equation*}
  [T_d,T_e]=f_{de}{}^fT_f,\qquad
  f^{ab}{}_{c}=\kappa^{ad}\kappa^{be}f_{de}{}^f\kappa_{fc}.
\end{equation*}
This agrees with the component convention in \eqref{eq:generic-zero-current-bracket}. Therefore
\begin{align}
  A_aB_b f^{ab}{}_{c}\widetilde j^c
  &=(\kappa^{da}A_a)(\kappa^{eb}B_b)
     f_{de}{}^f\kappa_{fc}\widetilde j^c\nonumber\\
  &=\pair{[A^\sharp,B^\sharp]}{\widetilde j}
   =\pair{[\grad P(j),\grad R(j)]}{\widetilde j}.
  \label{eq:app58-contraction}
\end{align}
If the symmetric tensors are substituted explicitly, the coefficient in the first line is
\begin{equation*}
  mn\,d^{(m)}_{a a_2\cdots a_m}
      e^{(n)}_{b b_2\cdots b_n}f^{ab}{}_{c}
      j^{a_2}\cdots j^{a_m}j^{b_2}\cdots j^{b_n}
      \widetilde j^c.
\end{equation*}
Equation \eqref{eq:app58-gradient-commutator} shows that this expression vanishes at every spatial point. Combining \eqref{eq:app58-delta-integration} and \eqref{eq:app58-contraction} gives
\begin{equation}\label{eq:app58-final}
  \{Q_P,Q_R\}_{\delta}
  =\alpha\circleint
      \pair{[\grad P(j(x)),\grad R(j(x))]}{\widetilde j(x)}\,\dd x
  =0,
\end{equation}
This is \eqref{eq:Evans-delta-part}.

The cancellation occurs before integration, so it is independent of the sign of $\alpha$, the relation between $\widetilde j$ and $j$, and the vanishing of periodic total derivatives. The complex semisimple case is included because the argument uses only nondegeneracy and invariance of the bilinear form and adjoint invariance of the polynomials. The remaining non-ultralocal contribution is
\begin{equation*}
  \{Q_P,Q_R\}_{\delta'}
  =\beta\circleint A_a(x)\partial_x B^a(x)\,\dd x,
\end{equation*}
Its vanishing requires the Evans compatibility discussed in the main text; equation \eqref{eq:Evans-delta-part} alone does not yet imply that the full charge bracket vanishes.

\end{document}